\documentclass[12pt]{article}
\usepackage{amsmath,amsthm,amssymb}
\usepackage{graphicx}
\usepackage{enumerate}
\usepackage{natbib}
\usepackage{url} 
\usepackage[utf8]{inputenc}
\usepackage{xcolor}

\newtheorem{theorem}{Theorem}
\newtheorem{corollary}{Corollary}
\newtheorem{assumption}{Assumption}

\newcommand{\blind}{1}

\renewcommand{\baselinestretch}{1.75}

\newcommand{\Hset}{\mathcal{H}}
\newcommand{\bY}{\mathbf{y}}
\newcommand{\prob}{\mathbb{P}}

\newcommand{\cpos}{\mathcal{Y}}
\newcommand{\boldy}{\mathbf{Y}}
\newcommand{\bX}{\mathbf{X}}
\newcommand{\bx}{\mathbf{x}}
\newcommand{\epos}{\mathcal{E}}
\newcommand{\apos}{\mathcal{E}^c}
\newcommand{\Phat}{\widehat{\prob}}

\theoremstyle{plain}
\newtheorem{definition}{Definition}

\begin{document}

\def\spacingset#1{\renewcommand{\baselinestretch}%
{#1}\small\normalsize} \spacingset{1}


\if1\blind
{
  \title{Modeling Data Analytic Iteration With\\ Probabilistic Outcome Sets}
  \author{Roger D. Peng\thanks{
    The authors gratefully acknowledge Lucy D'Agostino McGowan, Tiffany Timbers, and Alyssa Columbus for their constructive feedback and suggestions to improve this manuscript.}\hspace{.2cm}\\
    Department of Statistics and Data Sciences, University of of Texas, Austin\\
    and \\
    Stephanie C. Hicks \\
    Department of Biostatistics and Department of Biomedical Engineering \\   
    Johns Hopkins University}
  \maketitle
} \fi

\if0\blind
{
  \bigskip
  \bigskip
  \bigskip
  \begin{center}
    {\LARGE Modeling Data Analytic Iteration With \\  Probabilistic Outcome Sets}
\end{center}
  \medskip
} \fi

\bigskip
\begin{abstract}
In 1977 John Tukey described how in exploratory data analysis, data analysts use tools, such as data visualizations, to separate their expectations from what they observe. In contrast to statistical theory, an underappreciated aspect of data analysis is that a data analyst must make decisions by comparing the observed data or output from a statistical tool to what the analyst previously expected from the data. However, there is little formal guidance for how to make these data analytic decisions as statistical theory generally omits a discussion of who is using these statistical methods. In this paper, we propose a model for the iterative process of data analysis based on the analyst's expectations, using what we refer to as expected and anomaly \textit{probabilistic outcome sets}, and the concept of statistical information gain. Here, we extend the basic idea of comparing an analyst's expectations to what is observed in a data visualization to more general analytic situations. Our model posits that the analyst's goal is to increase the amount of information the analyst has relative to what the analyst already knows, through successive analytic iterations. We introduce two criteria--\textit{expected information gain} and \textit{anomaly information gain}--to provide guidance about analytic decision-making and ultimately to improve the practice of data analysis. Finally, we show how our framework can be used to characterize common situations in practical data analysis.
\end{abstract}

\noindent%
\vfill

\newpage

\newpage
\spacingset{1.9} 
\section{Introduction}
\label{sec:intro}

An important goal of statistical theory is to characterize the behavior of data analytic tools in the presence of random variation in the data. These theoretical properties are derived by making some assumptions about the data generating mechanism and by analyzing the mechanisms of the analytic tools themselves.  The properties of statistical methods are typically presented in a manner that is independent of any specific person applying the methods. Therefore, two data analysts employing the same model on the same data (making the same assumptions) should agree about the method's properties, which is useful when objectively evaluating or comparing a set of methods. 

In the practice of data analysis, data analysts apply numerous tools to data and must be familiar with their theoretical properties. In addition, analysts make decisions in response to observing the data or the output obtained from applying statistical and computational methods to the data. The decisions made by the analyst drive an iterative analytic process that ultimately leads to new knowledge or evidence regarding a scientific, business, or policy question~\citep{kimhardin2021, peng2022perspective}. It is natural to assume that different analysts, over the course of an entire analysis, might make different decisions regarding what tools to apply to the data and how to react to the output the tools generate~\citep{breznau2022observing}. These decisions may be influenced by the individual analyst's (a)~familiarity with data collection procedures;  (b)~knowledge of the properties of statistical techniques; and (c)~experience with computing, all of which might affect the analyst's expectations for observing the output of a given tool. This sense-making process is critically dependent on the analyst's perspective and previous understanding of the underlying phenomena being studied and the data~\citep{grolemund2014cognitive}. 

For example, an analyst who fits a linear regression model and obtains a negative slope when they were expecting a positive slope can react to that output in a variety of ways. An experienced practitioner might be familiar with possible reasons for why this can occur (such as outliers) and choose a diagnostic tool as a next step. A novice practitioner might not be familiar with the possibility of outliers and might choose to accept the result as part of the expected variation in the data. It may not be clear who is correct in this situation, but it is clear that different analysts can take different paths at this juncture.  

Others have noted that there is relatively little guidance about how to carry out a data analysis or how to make data analytic decisions~\citep{unwin2001patterns,wild1994embracing,peng2017commenton}. A contributor to this reality is that statistical theory generally omits a discussion of \textit{who} is using these statistical methods and applying them to the data. Such an omission allows for an objective discussion of statistical methods and their properties, but it also omits any information or background that an individual analyst might bring to a problem. Bayesian approaches allow for individual contributions via prior distributions, but only in the statistical modeling aspects of an analysis. Other important areas of data analysis, such as data processing, wrangling, tidying, and exploration have no corresponding framework for guiding the analyst forward.

The iterative nature of data analysis, which is driven forward by the numerous decisions made by the analyst, is another characteristic of practical data science that is often missing from statistical theory frameworks~\citep{peng2021diagnosing,peng2022perspective}. Real-world data analyses will typically involve long sequences of data operations, some of which may involve statistical models and some of which may involve basic data wrangling operations. At the end of each iteration, the analyst must decide what to do next, if anything: What tool to apply, what new information to gather, or perhaps what expertise to consult. While statistical theory is helpful in the areas involving the application of statistical models, what is largely missing is a general framework for describing how to make decisions when analyzing data. While previous work has suggested the outlines of what this framework might look like \citep{wild1999statistical, grolemund2014cognitive,greenhouse2018teaching,lovett2000applying}, details of how to operationalize it, in practice for analyzing data, are sparse.

John Tukey wrote in his landmark book \textit{Exploratory Data Analysis} that ``The greatest value of a picture is when it \textit{forces} us to notice what we never expected to see."~\citep{tukey1977exploratory}. While his comment was in reference to data visualization in particular, it is possible to take the ideas embedded in his comment, such as the dichotomy between expectations and observations, to build a more general framework for data analytic iteration. In this paper, we present a model for characterizing the iterative application of analytic tools to data and the information that is obtained. Our approach incorporates the unique perspective of the analyst conducting the analysis by considering what the analyst expects to observe before applying a tool to the data. A key question that we aim to address here is how does an analyst decide what to do next in an analytic iteration?

We propose a model for the iterative process of data analysis based on the analyst's expectations and the concept of statistical information gain. Our model posits that the analyst's goal is to increase the amount of information the analyst has, relative to what the analyst already knows, through successive analytic iterations. We introduce the notion of \textit{probabilistic outcome sets} to formalize what the analyst expects to see in an analysis output and what would be considered unexpected. We then propose two criteria---expected information gain and anomaly information gain---for modeling how an analyst might choose the next step in an analysis. Finally, we discuss how this model can be used to describe common data science practices and outline some of the possibilities of using this model in more complex data analyses.

\section{Data Generation Mechanisms in Analytic Iteration}
\label{sec:generationmechansims} 

Analytic iteration is defined here as the process by which an analyst (a)~plans a data operation by selecting an analytic tool (e.g. calculate the empirical mean) to apply to a given dataset, (b)~sets their expectations about the potential outcomes that could arise from the data operation, (c)~executes the data operation, (d)~compares the observed result to what the analyst expected prior to executing the operation, and (e)~decides which tool to use in the next iteration. Each step in the analytic iteration is designed to produce information via output such as numerical summaries, graphical summaries, or fitted models. 

The analyst's general goal in a data analysis is to learn about the  mechanism that produced the data while asking some scientific, policy, or business question related to that dataset. The observed data $\mathbf{X}$ are sampled from population $\mathcal{X}$ according to distribution~$F$ (which is assumed to be unknown). The analyst's best estimate of $F$ at a specific moment in an analytic iteration is called $\widehat{F}$. We assume the densities for these distribution functions exist, so that $f$ is the density associated with $F$ and $\hat{f}$ is the density associated with $\widehat{F}$. 

While a data analyst can examine the raw data directly, typically such examination will occur through the output of some tool $T$ applied to $\bX$. For example, the distribution of a variable may be examined by looking at a histogram. The output of $T$ is $\boldy=T(\bX)$ and we will be concerned with the probability distribution of this output, which will be denoted by $\prob$. Let the set $A$ represent a collection of possible values for $\boldy$~(we will be more precise about this set in Section~\ref{sec:outcomesets}). The probability of the event that the output $\boldy$ falls in $A$ is denoted as $\prob(\boldy\in A)\stackrel{\Delta}{=}\prob(\bX\in T^{-1}(A))=\int _{T^{-1}(A)}\,dF$. Here, $T^{-1}$ is the inverse mapping of $T$ and $T^{-1}(A)$ represents the various configurations of $\bX$ that result in $\boldy$ when $T$ is applied to $\bX$. The analyst's best estimate of the probability that $\boldy$ falls in a set $A$ is denoted as $\Phat(\boldy\in A)\stackrel{\Delta}{=}\Phat(\bX\in T^{-1}(A))=\int_{T^{-1}(A)}\,d\widehat{F}$.

We propose that analytic iteration is a sequential process by which the analyst learns about $F$ while asking some question, and investigating that question using the observed data $\mathbf{X}$ from $F$. This process is constrained by the requirements of the scientific, policy, or business question being asked. By asking a sequence of questions and applying a corresponding sequence of informative analytic tools to the data, the analyst gains information about $F$ at each step of the iteration. The goal for the analyst is for $\widehat{F}$ to be close enough to $F$ to serve the analyst's interests as well as the interests of any other stakeholders.

For the remainder of the paper we will assume that we are in the middle of a single iteration of a larger data analysis. We will use $\widehat{F}$ to indicate the analyst's \textit{current} best estimate of $F$ and $\Phat$ to indicate the analyst's current best estimate of $\prob$, the probability distribution associated with the output of tool $T$. While it might be more specific to denote $\Phat_T$ instead of $\Phat$, given that $\Phat$ is defined conditional on the tool $T$ being chosen, we omit the subscript in the remainder for the sake of brevity and because in a single iteration, the tool $T$ will be fixed.

\section{Informative Tools for Analytic Iteration}
\label{sec:tools}

At each stage of an analytic iteration, an analyst applies a tool $T$ to the data $\bX$ to produce output $\boldy$. The tool $T$ could range from a basic computation, such as computing the sample mean of a column in a table, to a complex statistical model. An initial problem to address is determining whether a given tool $T$ is informative for a given problem when it is applied to $\bX$. In the following section, we formally describe what is meant by an informative tool. 

For a given question or problem, we begin by considering a collection of competing hypotheses in $\Hset$ about $\mathbf{X}$, which we refer to as the hypothesis set. For example, the hypotheses in the hypothesis set could be defined on whether there are no missing values ($H_1$) or there are some missing values ($H_2$) in a datasest $\bX$. Similarly, a set of hypotheses could be defined on whether the data $\bX$ follow a symmetric ($H_1$) or right-skewed ($H_2$) distribution. Alternatively, a set of hypotheses could involve a statement about a population parameter~$\theta$, where one hypothesis is $H_1: \theta=\theta_1$ for some $\theta_1$, a second hypothesis is $H_2: \theta = \theta_2$, and third hypothesis is $H_3: \theta=\theta_3$.

For simplicity, in the definition of an informative tool below, we consider the scenario where the hypothesis set $\Hset$ only contains two competing hypotheses $\Hset = \{H_1, H_2\}$. The basic concept here is that tool is an informative tool if the expected outcome of applying that tool to the data is different under $H_1$ relative to $H_2$. 

\begin{definition}[Informative tool in analytic iteration]
\label{def:informative}
A tool $T$ applied to a given dataset $\bX\sim F$ is an informative tool for discriminating between competing hypotheses $H_1$ and $H_2$ if 
$\mathbb{E}_{F_{H_1}}[T(\bX)] \neq \mathbb{E}_{F_{H_2}}[T(\bX)]$.
\end{definition}

The property of informativeness is not an inherent or intrinsic property of any given data analytic element or tool. Informativeness is something that can only be determined in the context of the hypothesis set $\Hset$ under consideration. We cannot say whether a boxplot or histogram is inherently informative or not for all types of questions asked about $\mathbf{X}$. For example, a histogram may be informative for discriminating between hypotheses about the symmetry or skewness of the data. However, the histogram will not be informative for discriminating between hypotheses concerning the correlation of a variable with another variable. For more general hypothesis sets $\Hset$, which may for example contain an infinite number of hypotheses, an informative tool is expected to produce distinct output from the data under each distinct hypothesis. 
In the remainder of this paper, we will assume that the tools applied to the data are informative in the sense of Definition~\ref{def:informative}.

\section{Potential Outcome Sets for Analytic Iteration}
\label{sec:outcomesets}

We begin by assuming a data operation has been planned, so that the analyst has chosen a tool to apply to a dataset. In this section, we define three sets of events that could occur related to the potential (or possible) outcome of applying a given analytic tool to a dataset in a single analytic iteration. We describe three \textbf{potential outcome sets} (Definitions \ref{def:comp-pot-outcome}-\ref{def:anom-pot-outcome}) that characterize events \textit{before} the analyst applies a given analytic tool to the data. 

\subsection{Complete Potential Outcome Set}

Once an analyst has chosen an analytic tool to apply to the data, it is usually possible to characterize in some way the set in which the output might fall, before the tool is actually applied to the data. In the following, we consider the unobserved data as a random variable~$\bX$. For a given tool, $T$ we aim to describe the entire set of possible outcomes from the tool when applied to $\bX$. We call this set the \textbf{complete potential outcome set}. 

\begin{definition}[Complete potential outcome set]
    Let $\bX\in\mathcal{X}$ be a random variable with a cumulative distribution function $F$. We consider $\mathbf{Y} = T(\bX)$, defined as a mapping from the original sample space of $\bX$, 
    $\mathcal{X} = \{\bx: f(\bx) >0 \}$, to a new sample space $\mathcal{Y}$, the sample space for the random variable $\boldy$, i.e. $T: \mathcal{X}\mapsto \mathcal{Y}$.
    We associate $T$ with an inverse mapping, denoted by $T^{-1}$, which maps subsets of $\mathcal{Y}$ to subsets of $\mathcal{X}$, and is defined by 
    $$T^{-1}(A) = \{\bx \in \mathcal{X}: T(\bx) \in A \}.$$
    We also assume the mathematical domain (or input) of $T$ is valid for the data type $\bX$. Then, the sample space $\mathcal{Y} = \{\bY: \bY = T(\bx) \text{ for some } \bx \in \mathcal{X} \}$, which is the image of $\mathcal{X}$ under $T$, is the \textbf{complete potential outcome set} that describes all the potential outcomes of the analytic step for applying tool $T$ to the data $\bX$.
    \label{def:comp-pot-outcome}
\end{definition}

For example, consider a set of observations $\bx = (x_1, \ldots, x_n)$ that comes in the form of a collection of $n$ non-negative integers with sample space $\mathcal{X} = \mathbb{Z}^{+}_{0} = \{0, 1, \ldots, \}$. If we consider the sample mean as a tool $T$, $\bY = T(\bx) = \frac{1}{n}\sum_{i=1}^n x_i$, then the complete potential outcome set is the positive half of the real line $[0, \infty)$, or $\mathcal{Y} = \{\bY: \bY=T(\bx), \bx\in\mathcal{X} \} = \mathbb{R}^{+}_0$.

For analytic tools that summarize numerical data, such as calculating a sample mean, it is straightforward to specify their complete potential outcome sets with some precision. However, with other output, such as graphical or high-dimensional output, it may be more difficult to specify the complete set of potential outcomes. Nevertheless, it may be possible to describe in a general sense what output could be produced, and there is value in developing such a description. For example, when making a scatter plot of pairs of $n$ data points $(x_{11},x_{21}),\dots,(x_{1n},x_{2n})$, where $x_{1i}$, $x_{2i}$ $\in \mathbb{R}$, then the complete potential outcome set is the set of all possible combinations of $n$ points arranged on the plane $\mathbb{R}^2$.

The complete potential outcome set follows from the nature of the data and the properties of the analytical tool. Therefore, this set is common to any analysis meeting the same specifications. An assumption that we make is that most analysts would agree on the characteristics of the complete potential outcome set, conditional on the data and the chosen analytic tool. Thus, the complete potential outcome set does not yet incorporate an analyst's unique perspective.

\subsection{Probability Structure for Outcome Sets}
\label{sec:probability}

In order to further model the analytic process, we need to introduce a probability structure for modeling potential outcome sets. This probability structure will allow us to introduce and model two additional outcome sets below, namely the expected and anomaly potential outcome sets. The probabilities here characterize the likelihood of the output $\mathbf{Y}$ from the analytic tool $T$ on the data $\bX$ falling in a given potential outcome set. Because the true data generating mechanism is presumably not known, the analyst will have to estimate the probability that the outcome of an analytic tool falls anywhere in a potential outcome set. To distinguish between the analyst's estimate of the data generating mechanism and the true mechanism, we denote $\Phat(A)$ to indicate the analyst's estimate rather than $\prob(A)$ for some event $A$, as described in Section~\ref{sec:generationmechansims}. 

The randomness that gives rise to the probability structure that we propose has two primary sources. The first is the traditional statistical variation that would cause two different sampled datasets to be different in their values. The second source of variation arises from ``upstream uncertainty", which is the series of activities that occurs between data collection and the generation of data analytic output. For example, in some contexts, batch effects can be introduced in data collection that can affect inferences in downstream analyses~\citep{leek2010tackling}. Similarly, data recording errors can produce outliers in certain variables. In both cases, the analyst may not initially be aware of these features of the data. If the analyst is not able to completely characterize these upstream processes, then it may be reasonable to consider this aspect of the data generation as random. 

Because of these various sources of randomness, the analyst's estimated probability $\Phat$ could be constructed using a variety of mechanisms. The analyst could use a formal probability model with hypothesized parameter values or prior distributions. However, in other less structured situations, it may not be feasible to employ a formal model and the analyst may rely on more coarse specifications of the probability distribution. For example, in estimating the mean of a random variable, one analyst may feel confident in using a fully specified Normal distribution to predict the value of the sample mean while another analyst (perhaps less experienced with this type of data) may only feel comfortable specifying the probability of the sample mean falling in broadly defined sets.

Once the analytic step has been taken and the method is applied to the data, we observe the result and there is nothing random. Hence, all discussions of probability in this section are to be considered in the context prior to applying a method to the data. However, the probabilities assigned to the potential outcomes will play an important role in determining what information is gained upon observing the result~(Section~\ref{sec:analyticstep}).

\subsection{Complete Potential Outcome Set} 

Next, we describe three assumptions about complete potential outcome sets. First, a basic, but also an important assumption that we make regarding the data $\bX$, the tool $T$, and the output $\boldy$, is that the complete potential outcome set $\cpos$ contains all of the outcomes that we could observe. 
\begin{assumption}[Complete potential outcome set contains all outcomes]
Suppose the analyst has chosen $T$ to apply to the data $\bX$ and let $\boldy$ represent the output from that method (as yet unobserved). Let $\cpos$ represent the complete potential outcome set for $\boldy$ and let $\{\boldy\in\cpos\}$ represent the event that the output $\boldy$ is observed to fall in $\cpos$. Let $\prob$ represent the true, but unknown, probability distribution governing the output $\boldy$. Then, $\prob(\boldy\in\cpos)=1$.     
\label{assumption-cpos1}
\end{assumption}

The second assumption we make states that although an analyst cannot know $\prob$, the probability distribution governing the output $\boldy$, the analyst's best estimate of $\prob$, i.e. $\Phat$, satisfies the same assumption.    
\begin{assumption}[Complete potential outcome set contains all of analyst's outcomes]
Let $\Phat$ be an analyst's best estimate of the true probability $\prob$, which is unknown in general. Then, we also assume that $\Phat(\boldy\in\cpos)=1$.  
\label{assumption-phatcpos}
\end{assumption}

Third, we note that $\prob$ and the complete potential outcome set $\cpos$ are universal in the sense that they are independent of the analyst looking at the data. However, $\Phat$ is specific to an analyst and will likely differ from analyst to analyst. 

\begin{assumption}[Universality of the complete potential outcome set]
Suppose two analysts labeled A and B look at the same dataset $\bX$ independently of each other and apply the same tool $T$ to produce output $\boldy$. Denote $\Phat_A$ and $\Phat_B$ as the estimates of $\prob$ for analysts A, and B, respectively. Then, we assume that $\Phat_A(\boldy\in\cpos) = \Phat_B(\boldy\in\cpos)=1$.     
\end{assumption}

While it is likely that $\Phat_A\ne\Phat_B$, so that two different analysts might disagree on how the probability mass of the output is distributed across $\cpos$, they should still agree on the overall structure of $\cpos$ and the boundaries of where $\boldy$ might fall, regardless of how likely or unlikely. In what follows, we will use $\Phat$ to denote an arbitrary analyst's estimate of the probability distribution.

Continuing the example from the previous section with non-negative integer data $\bX=(X_1,\dots,X_n)$, if our analytic method $T$ were the sample mean such that $\boldy=T(\bX)$, then our complete potential outcome set is $\cpos = \mathbb{R}_0^+$. Under this setup, it would be reasonable to assume that $\prob(\boldy\in\cpos)=1$ and that an analyst's best estimate $\Phat$ would also satisfy $\Phat(\boldy\in\cpos)=1$.

Next, we partition the complete potential outcome set $\cpos$ into two subsets, namely the expected potential outcome set (Section~\ref{sec:expectedoutcome}) and the anomaly potential outcome set (Section~\ref{sec:anomalyoutcome}).

\subsection{Expected Potential Outcome Set}
\label{sec:expectedoutcome}

In this section, we introduce the \textbf{expected potential outcome set}, which is a subset of the complete potential outcome set and is unique to the data analyst. In our analytic iteration model, before applying $T$ to the data, the analyst often makes a (explicit or implicit) prediction as to what the output $\boldy=T(\bX)$ is expected to be, or the range that the output will fall in. Such predictions can range from being vaguely characterized to being very well specified, perhaps based on some sort of previously developed hypothesis or preliminary data. 

The expected potential outcome set $\epos\subset\mathcal{Y}$ is the set, according to the best judgment of the analyst, in which the output $\boldy=T(\bX)$ is expected to fall, which is a subset of the complete potential outcome set $\cpos$. Specifying the expected potential outcome set is important for interpreting the output from the analytic step and for deciding on any next steps in the analytic iteration. In particular, the expected potential outcome set indirectly specifies which outcomes are \textit{unexpected}, which we will discuss in the next section. 

\begin{definition}[Expected potential outcome set]
    Let $\boldy=T(\bX)$ be the output of applying an analytic tool $T$ to the data $\bX$ and let $\epos$ be a subset of the complete potential outcome set $\cpos$. $\epos$~is an \textbf{expected potential outcome set} for $\boldy$ if it satisfies:
    \begin{enumerate}
        \item $\Phat(\boldy\in\epos) > 0.5$; and
        \item $\Phat(\boldy\in\epos) < 1$,
    \end{enumerate}
    where $\Phat$ is the analyst's best estimate of the true probability structure $\prob$.
    \label{def:exp-pot-outcome}
\end{definition}

A defining characteristic of $\epos$ is that the analyst believes that the outcome is expected to fall in that set, by which we mean that $\boldy$ is ``more likely than not" to fall in $\epos$, in the best judgment of the analyst. 

The expected potential outcome set depends on many factors, including the analyst's current state of knowledge of the science, assumptions about the underlying data generation process, the ability to model the data using statistical or computational tools, knowledge of the theoretical properties of those tools, and the understanding of the uncertainty or variability of the observed data across multiple samples. Hence, even if the underlying truth is thought of as fixed, it is reasonable to assume that different analysts might develop different sets of expected potential outcomes, reflecting differing levels of familiarity with the various factors involved and different weightings of existing evidence or data.

\subsubsection{Example: Reading in a data file}
\label{sec:read-data-file}

Consider a scenario where an analyst receives a data file and is told by the project scientist that the data file contains 1,000 observations. The analyst's first step might be to read in the data and verify that the resulting data object is a table with 1,000 rows. Here, the tool $T$ might be a function that counts the number of rows or perhaps a printed summary of the number of rows provided by the software. The ``data" $\bX$ in this case is the dataset file being read in, and the output $\boldy$ is the number of rows. The complete potential outcome set $\cpos$ is the set of non-negative integers and the expected potential outcome set might be $\epos=\{1000\}$, i.e. the set containing a single integer. Finally, the analyst might estimate that $\Phat(\boldy\in\epos)=0.99$ based on the fact that the project scientist confirmed beforehand that there were 1,000 rows.

\subsection{Anomaly Potential Outcome Set}
\label{sec:anomalyoutcome}

The \textbf{anomaly potential outcome set} contains potential outcomes for $\boldy$ that are \textit{unexpected} by the analyst for data $\bX$ using analytic tool $T$. The outcomes characterized by the anomaly potential outcome set are not impossible, but are deemed unlikely to occur in the eyes of the analyst. The anomaly potential outcome set will typically not be specified directly, but rather indirectly via the specification of the complete potential outcome set and the expected potential outcome set. 

\begin{definition}[Anomaly potential outcome set]
    Given a complete potential outcome set $\cpos$ and expected potential outcome set $\epos\subset\cpos$, the \textbf{anomaly potential outcome set} is $\apos=\cpos\setminus\epos$.
    \label{def:anom-pot-outcome}
\end{definition}

Using the definitions of $\cpos$ and $\epos$, it follows that for an analytic output $\boldy$, $\Phat(\boldy\in\apos) > 0$ and $\Phat(\boldy\in\apos) = 1-\Phat(\boldy\in\epos) < 0.5$. As $\apos$ is the complement of $\epos$ in $\cpos$, the interpretation of $\apos$ is that the output $\boldy$ may fall in $\apos$, but is ``less likely than not" to do so.

Continuing the data file example from the previous section, because $\epos=\{1000\}$ and $\cpos=\mathbb{R}^+_0$, then the anomaly set is $\apos=\{0,1,\dots,999,1001,1002,\dots\}$. Given that the analyst specified $\Phat(\boldy\in\epos)=0.99$, we have that $\Phat(\boldy\in\apos)=0.01$. Hence, from the perspective of the analyst, it is highly unlikely that the number of rows presented in the output will be anything other than 1,000.

\subsection{From Outcome Sets to Information}

Dividing the complete potential outcome set into the complementary expected and anomaly potential outcome sets is a task that separates data analysis problems from traditional statistical problems. In particular, this division must be done by a person, namely the data analyst, and every analyst will likely devise a different configuration. The expected and anomaly potential outcome sets are dependent on the analyst's judgment and are not inherent aspects of the science or the data. Therefore, as the analyst's information about the data and the data generating mechanism accumulates, it is possible that the expected potential outcome sets for various outputs will evolve.

Developing the different potential outcome sets in an analytic step, as well as their probability structure, allows us to characterize what we see as the goal of data analysis, which is to gain information. Information cannot be obtained without risk, and in our model, information is gained from resolving the uncertainty about the output $\boldy$ obtained from applying $T$ to $\bX$. By observing the result of applying an analytic tool $T$ and seeing whether it falls in $\epos$ or $\apos$, this leads to a formal framework for how to determine what is learned in a step of an analytic iteration.

\section{Information Gain From A Single Analytic Step}
\label{sec:analyticstep}

Upon observing $\boldy$, the output from a single analytic step, we can determine whether one of the mutually exclusive events $\{\bY\in\epos\}$ or $\{\bY\in\apos\}$ has occurred. 
We first define the information that is gained from observing the output from an analytic step. We then define information-based criteria that can be used to characterize the available tools based on an analyst's goals or current knowledge of the data.

\subsection{Information Gain}
\label{sec:informationgain}

For a given tool $T$, we describe below how the information gained from applying $T$ to the data $\bX$ is determined by the expected and anomaly potential outcome sets $\epos$ and $\apos$ constructed by the analyst and $\Phat$, the analyst's best estimate of the probability distribution for the output $\boldy$. Because $\epos$ is defined by the analyst using $\Phat$, the information gained from observing the potential outcome is similarly defined in terms of~$\Phat$. First, we introduce the quantitative characterization of information in this context, which is equivalent to Shannon information \citep{shannon}. Then, we introduce properties of information in this context. 

\begin{definition}[Information gain from an analytic step]
    Let $\bX$ be a dataset and $T$ be any analytic tool that can be applied to $\bX$ to produce output $\boldy=T(\bX)$. Let $S$ be any subset of the complete potential outcome set $\cpos$. Then, the \textbf{information gained} by the analyst from observing the event $\{\boldy\in S\}$ in a single analytic iteration step is $-\log\Phat(\boldy\in S)$.
    \label{def:infogain}
\end{definition}

The information structure of an analytic step is introduced by the analyst via the specification of~$\epos$ and~$\apos$ within~$\cpos$. It follows from Section~\ref{sec:probability} Assumption~\ref{assumption-cpos1} (namely complete potential outcome set $\cpos$ contains all outcomes) and the definitions of $\epos$ and $\apos$ in Sections~\ref{sec:expectedoutcome}-\ref{sec:anomalyoutcome} that the outcome $\boldy$ must fall into one of those two sets. If we observe the event $\{\boldy\in\epos\}$, we say the result is as-expected; if we observe the event $\{\boldy\in\apos\}$, we say the result is unexpected. Note that without the specification of $\epos$ and $\apos$, there is no possibility of information gain. 

\begin{theorem}[No information gain without specifying expected potential outcome set]
Without any specification of the expected potential outcome set $\epos$ (and the anomaly potential outcome set $\apos$), the information gained from an analytic step is $0$.
\end{theorem}

\begin{proof}
If all that is specified is that $\boldy$ can fall in the complete potential outcome set $\cpos$, then it follows from Assumption~\ref{assumption-phatcpos} (namely, $\Phat(\boldy\in\cpos)=1$) and Definition~\ref{def:infogain} that the information gained from that outcome is $-\log\Phat(\boldy\in\cpos)=-\log 1 =0$. 
\end{proof}

By Definition~\ref{def:infogain}, the information gained by the analyst upon observing $\{\boldy\in\epos\}$ or $\{\boldy\in\apos\}$ is then $-\log\Phat(\boldy\in\epos)$ or $-\log\Phat(\boldy\in\apos)$, respectively.  The next property shows that no matter what potential outcome set the output~$\boldy$ ultimately falls in, there will be some non-zero information gain.

\begin{theorem}[Information positivity]
Upon specification of the expected potential outcome set $\epos$ (and the anomaly potential outcome set $\apos$), the information gained from an analytic step is strictly positive and finite.
\end{theorem}

\begin{proof}
By Definitions~\ref{def:exp-pot-outcome} and \ref{def:anom-pot-outcome}, we know that $0 <\Phat(\boldy\in\epos) < 1$~ and $0 <\Phat(\boldy\in\apos) < 1$, respectively. Therefore, it follows that $-\log\Phat(\boldy\in\epos) > 0$ and $-\log\Phat(\boldy\in\apos) > 0$. 
\end{proof}

The next property explains how two different analysts applying the same tool to the same dataset can have different information gain upon observing the same output.

\begin{theorem}[Different information gain for multiple analysts]
Suppose we have two analysts, analyst $A$ and analyst $B$, and each has chosen to apply tool $T$ to the data~$\bX$~(with complete potential outcome set $\cpos$). Let $\epos_A$ and $\epos_B$ be the expected potential outcome sets in $\cpos$ for analysts $A$ and $B$, respectively, and let $\Phat_A$ and $\Phat_B$ be their corresponding probability densities for the output~$\boldy$. If $\Phat_A(\boldy\in\epos_A)\ne\Phat_B(\boldy\in\epos_B)$,
then the two analysts will gain different amounts of information upon observing the output~$\bY$.  
\label{theorem:different-info}
\end{theorem}

\begin{proof}
    See Supplementary Note~\ref{sec:suppnote1}
\end{proof}

\begin{corollary}
    If two analysts $A$ and $B$ experience the same information gain from observing the output $\bY$, then it follows that $\Phat_A(\boldy\in\epos_A)=\Phat_B(\boldy\in\epos_B)$ and that the observed outcome was either as-expected for both analysts or unexpected for both analysts.
\label{corollary:different-info}
\end{corollary}

\begin{proof}
    See Supplementary Note~\ref{sec:suppnote2}
\end{proof}

Both Theorem~\ref{theorem:different-info} and Corollary~\ref{corollary:different-info} highlight the fact that two analysts observing the same output generated by the same tool applied to the same dataset can experience different information gains depending on what their expectations were for the output and how their respective probability models differed for~$\boldy$.

\subsection{Expected Information Gain}
\label{sec:expected-info-gain}

When evaluating different analytic tools to apply to the data, one consideration is the amount of information that we expect to obtain from applying the tool to the data regardless of whether the outcome is as-expected or unexpected. First, we can define the observed information gained from applying a tool $T$ to the data $\bX$ in a single analytic step. 

\begin{definition}[Observed information gain from an analytic step]
    Let $\bX$ be a dataset and $T$ any analytic tool to be applied to $\bX$ to produce output $\boldy=T(\bX)$. Then the \textbf{observed information gained} by the analyst after observing $\boldy$ is
$$
G = 
\mathbf{1}\{\boldy\in\epos\}\left[-\log\Phat(\{\boldy\in\epos\})\right]
+
\mathbf{1}\{\boldy\in\apos\}\left[-\log\Phat(\{\boldy\in\apos\})\right].
$$
\label{def:observed-infogain}
\end{definition}

The value of $G$ cannot be computed until \textit{after} $T$ has been applied to $\bX$ and the output $\boldy$ has been observed, so it is not necessarily useful for planning purposes. However, we can compute its expectation with respect to $\Phat$, which is the average amount of information gained from applying a tool with a given expected potential outcome set.

\begin{definition}[Expected information gain from an analytic step]
Let $\epos$ and $\apos$ be the expected and anomaly potential outcome sets, respectively, associated with a tool $T$ and dataset $\bX$. Let $\boldy$ be the output obtained from applying $T$ to $\bX$. The \textbf{expected information gained} by the analyst from observing the output $\boldy$ is
\begin{eqnarray*}
\widehat{H} 
& = &
\mathbb{E}_{\Phat}[G]\\
& = &
\mathbb{E}_{\Phat}[\mathbf{1}\{\boldy\in\epos\}]\left[-\log\Phat(\boldy\in\epos)\right]
+
\mathbb{E}_{\Phat}[\mathbf{1}\{\boldy\in\apos\}]\left[-\log\Phat(\boldy\in\apos)\right]\\
& = &
\Phat(\boldy\in\epos)\left[-\log\Phat(\boldy\in\epos)\right]
+
\Phat(\boldy\in\apos)\left[-\log\Phat(\boldy\in\apos)\right].
\end{eqnarray*}
\label{def:expected-infogain}
\end{definition}

We note that the expectation is taken with respect to $\Phat$, which is the analyst's estimate of the probability distribution for the output, not the true distribution. Therefore, it is possible that the expected information gain computed here will diverge from what might have been computed under the true $\prob$. We will discuss this divergence further in Section~\ref{sec:analytic-progress}.

\subsection{Anomaly Information Gain}

The expected information gain $\widehat{H}$ characterizes the average information that would be obtained in a step regardless of whether the outcome is as-expected or unexpected. However, it might also be worthwhile to consider an alternate criterion that focuses on the information obtained solely when an unexpected outcome is observed.

\begin{definition}[Anomaly information gain from an analytic step]
Let $\epos$ and $\apos$ be the expected and anomaly potential outcome sets, respectively, associated with a tool $T$ and dataset $\bX$. Let $\boldy$ be the output obtained from applying $T$ to $\bX$. The anomaly information gained by the analyst from observing output $\boldy$ is
$$
\widehat{M}
= -\log\Phat(\boldy\in\apos).
$$
\label{def:anomaly-infogain}
\end{definition}

It is useful to note that because we define $\apos$ such that $\Phat(\boldy\in\apos)> 0$~(Definition~\ref{def:anom-pot-outcome}), $\widehat{M}$ is always finite. 

\begin{theorem}
$\widehat{M}$ represents the maximum possible information gain that can be obtained from applying the tool $T$ to the data $\bX$.
\end{theorem}

\begin{proof}
Because the $-\log$ function is monotonically decreasing on $(0,1)$ and Definitions~\ref{def:exp-pot-outcome} and~\ref{def:anom-pot-outcome} restrict $\epos$ so that $\Phat(\boldy\in\apos) < \Phat(\boldy\in\epos)$, it follows that $-\log\Phat(\boldy\in\apos)>-\log\Phat(\boldy\in\epos)$.
\end{proof}

\subsection{Example: Estimating the Sign of a Correlation}

A common data analytic task is to determine whether two variables are positively or negatively related to each other. For the sake of this example, we do not consider this task as an inferential one, but rather as a computational step that is perhaps part of a larger analysis. With two numerical variables $\bx_1=x_{11},\dots,x_{1n}$ and $\bx_2=x_{21},\dots,x_{2n}$, we can compute a Pearson correlation coefficient $r(\bx_1,\bx_2)$ and observe the sign of the output value. In this case, the tool we have chosen is $T= \text{sign}(r(\cdot,\cdot))$ and the output is $\boldy=\text{sign}(r(\bx_1,\bx_2))$. The complete potential outcome set for $\boldy$ here has three elements, $\cpos=\{1, 0, -1\}$.

An analyst may expect that the two variables are positively related and so the expected potential outcome set is $\epos=\{1\}$ and the anomaly potential outcome set is $\apos=\{0,-1\}$. Furthermore, the analyst may be quite confident about this problem and may assign the event $\{\boldy\in\epos\}$ a probability of $\Phat(\boldy\in\epos)=0.95$. (Note that for this example we will use $\log_2$ so that $\widehat{H}$ ranges from 0 to 1.) From the analyst's perspective, the expected information from this single analytic iteration is $\widehat{H}= 0.29$~bits. The expected information is relatively low here given the perceived high likelihood of the result being as-expected. Stated another way, if an analyst places high confidence on the event $\{\boldy\in\epos\}$, then the expected information gain is small. In contrast, the anomaly information gain in this scenario is $\widehat{M}=4.3$~bits, which is the information gained in the event that the observed sign is either $0$ or $-1$. Given the high confidence in the correlation being positive, we gain a substantial amount of information from observing an anomaly in this case.

\section{Applications}
\label{sec:applications}

The model for analytic iteration that we have developed in the previous sections aims to provide a formal framework to characterizing a broad range of data analytic steps that would be taken over the course of a data analysis. These steps include data wrangling, cleaning or processing, exploratory data analysis and visualization, and statistical modeling. The approach we have taken parallels the traditional statistical paradigm where an observed output is decomposed into a ``model" predicted value plus a ``residual" value. However in our proposed framework, we have decomposed each analytic step into expected and anomaly potential outcome sets. Next, we outline some applications of this framework to formally characterize some common data analysis activities.

\subsection{Defining One Data Analytic Iteration}
\label{sec:oneiteration}

First, we use this framework to address a question on `What is one data analytic iteration?'. We define a single data analytic iteration, based on (i) the specification of the two potential outcome sets, both of which characterize events \textit{before} the analyst applies the analytic tool to the data, and (ii) observing the output \textit{after} applying the tool to the data:

\begin{itemize}
    \item[(i)] The complete potential outcome set $\mathcal{Y}$ (\textbf{Definition~\ref{def:comp-pot-outcome}}), which describes all the potential outcomes of the analytic step for applying tool $T$ to the data $\bX$. Note, while choice of $T$ depends on the analyst, the specification of $\mathcal{Y}$ is only dependent on $T$.  
    \item[(ii)] The expected potential outcome set $\epos\subset\cpos$ (\textbf{Definition~\ref{def:exp-pot-outcome}}), which describes the set, according to the best judgment of the analyst, in which the output $\boldy=T(\bX)$ is expected to fall. Note, specification of $\epos$ is dependent on the analyst and also indirectly specifies the anomaly potential outcome set $\apos=\cpos\setminus\epos$ (\textbf{Definition~\ref{def:anom-pot-outcome}}). 
    \item[(iii)] Observing the output $\bY = T(\bx)$ for some $\bx$ and determining whether one of the mutually exclusive events $\{\bY\in\epos\}$ or $\{\bY\in\apos\}$ has occurred.
\end{itemize}

\begin{figure}[tbh]
    \centering
    \includegraphics[width=6in]{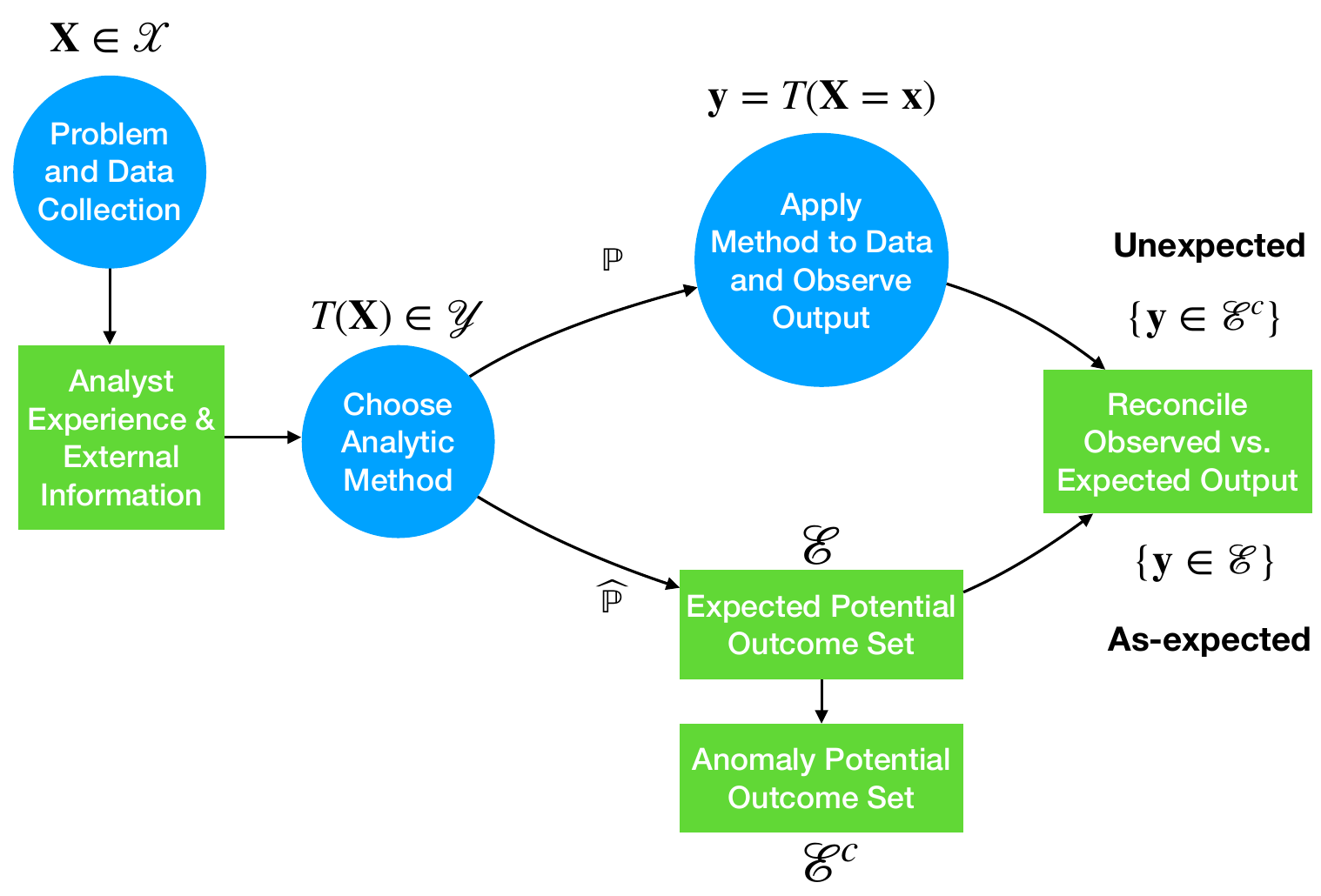}
    \caption{A single data analytic iteration involving choosing a tool $T$ with complete potential outcome set $\cpos$, applying the tool to the data $\bx$ to observe the output $\bY$, and determining whether the output $\bY$ falls in the expected potential outcome set $\epos$ or the anomaly potential outcome set $\apos$. The green rectangles highlight areas of the iteration where the analyst makes a contribution to the analytic process.}
    \label{fig:analyticiteration}
\end{figure}

\noindent The combination of these three components make define one \textit{data analytic iteration}~(see Figure~\ref{fig:analyticiteration}). The key insight here is defining the end of the data analytic iteration after observing the output $\bY$ and determining whether one of the mutually exclusive events $\{\bY\in\epos\}$ or $\{\bY\in\apos\}$ has occurred. The next step starts a decision making process about what to do next, using the observation that has just been realized. For example, the next analytic iteration begins where the analyst chooses amongst a collection of informative tools that we can apply to the data in the next iteration (\textbf{Section~\ref{sec:analyticstep}}).

\subsection{Single vs.~Multiple Iteration Analyses}
\label{sec:typesofanalyses}

Different types of data analyses may have different numbers of data analytic iterations carried out. Data from clinical trials are often analyzed using a detailed analytic protocol that is written before the data are collected. Once the data are collected, the pre-specified analysis protocol is executed without deviation. In this scenario there is essentially only a single analytic step. While that step may be complicated and statistically sophisticated, there are no intermediate results for the analyst to look at and consider, and no expectations for the analyst to propose. However, once the analysis is done and the results are in hand (or perhaps already published), the analyst may want to explore the data more to explain the observed results or to gain some further insight. This typically does not involve \textit{changing} the original analysis protocol. Rather, this next iteration of the analysis would add to what was already produced in order to provide a clearer explanation or interpretation to the results. A key distinction of this type of analysis is that there is typically a significant of preparation that is done to understand the data generation process or the data collection protocol. Therefore, uncertainties about those aspects of the data can be more tightly controlled.

Other, perhaps more exploratory, types of analyses may involve multiple iterations through the data in order to extract successive pieces of information. If the analyst is not familiar with the data collection protocol, and particularly if they were not in control of the data collection, then they analyst may need multiple iterations to learn better characterize the uncertainty in the data. Such exploratory analyses may be more common in secondary data analyses when there is a separation between those who collected the data and those who analyze it. In these situations, care must be taken to ensure that results are presented honestly and that uncertainty about any findings are accurately taken into account.

\subsection{Evaluating Analytic Progress}
\label{sec:analytic-progress}

In Section~\ref{sec:expected-info-gain} we described two numerical quantities, the observed information gain from an analytic step $G$, and the expected information gain $\widehat{H}$ (from the analyst's perspective). The value of $\widehat{H}$ is entirely determined by the analyst's estimate of the data generation mechanism, while the value of $G$ is determined in part by the true underlying data generation mechanism. Therefore, it may be of interest to compare the observed $G$ after a sequence of analytic steps to the analyst's stated value of $\widehat{H}$, particularly if the analyst is using $\widehat{H}$ to guide the choice of tools. If the analyst were to make their specification of $\epos$ and $\apos$ known, one could compare the observed information gain $G$ to the expected information gain $\widehat{H}$ to assess the reasonableness of the analyst's perspective. If $\Phat\approx\prob$, so that the analyst has a reasonable estimate of the true probability distribution governing $\boldy$, then $G$ and $\widehat{H}$ should be close on average. For a single iteration, any deviation between $G$ and $\widehat{H}$ might be difficult to interpret. However, with many steps executed over the course of an analysis, a pattern may emerge of large or small deviations.

Consider a sequence of analytic steps $t=1, 2, \dots$ where at each step the analyst must specify a value for $\widehat{H}_t$ based on their current understanding of the data generating mechanism and their choice of analytic tool. Upon executing the analytic step we observe $G_t$. Over the course of the sequence of analytic steps, we can compute the cumulative expected information $S_{\widehat{H}}=\sum_t\widehat{H}_t$ and the cumulative observed information gain $S_G=\sum_t G_t$. If the analyst has a good understanding of the data generating mechanism, it might be expected that $S_G$ and $S_{\widehat{H}}$ would be close.

Consider a scenario where an analyst is highly confident of the outcome and predicts it with 95\% probability. Should an anomaly be observed, so that $\boldy$ is observed to fall in the anomaly potential outcome set $\apos$, then this single observation might be attributed to random chance. However, over the course many analytic steps, if anomalies were to continue to occur, despite the analyst's confidence, then we would expect that $S_G$ would deviate substantially from $S_{\widehat{H}}$. In this scenario, we could potentially use the deviation between $S_G$ and $S_{\widehat{H}}$ as a measure of how far off is the analyst's understanding of the data.

\section{Discussion}
\label{sec:discussion}

Substantial previous work has described data visualizations as a tool in exploratory data analyses that could be used to guide data analysts on how to separate their expectations from what they observe in the practice of data analysis \citep{tukey1977exploratory, anscombe1973graphs}. In this paper, we extend the basic idea of comparing expectations to observations in data visualizations to more general analytic situations. Specifically, we propose a model for the iterative process of data analysis based on the analyst's expectations using what we refer to as \textit{probabilistic outcome sets} and statistical information gain. Each analytic step induces a set of potential outcomes which is then partitioned by the analyst into an expected potential outcome set and an anomaly potential outcome set. These sets are then assigned probabilities which can be used to quantify the information and evidence that can be obtained from the analytic step.

It is important to note that the concept of developing expectations and comparing them to observed data is fundamental to all of statistics and has previously been applied in the context of data analysis~\citep{grolemund2014cognitive}. Statistical models themselves are formal representations of an analyst's or scientist's expectations for the data. However, the revolution in data science in recent years has shined a light on the less formal, yet critically important, aspects of data analysis that have yet to admit a mathematical model. In fact, many have argued that data analysis is a skill that can be learned only through experience~\citep{greenhouse2018teaching}. If so, it should be unsatisfying to the statistical community that there is an important part of our field that can only be learned by experiencing it. It is only logical then to ask what is it that we learn through the experience of doing data analysis and what can we generalize across our accumulated experiences? 

Here, we argue that analytic iteration, the process of developing expectations, reconciling output with our expectations, and determining next steps, is a generalizable aspect of data analysis that deserves more attention. We propose a model for this cyclic process in order to provide some formal structure to the analytic process without being prescriptive. Furthermore, we introduce some language for describing this process when teaching it to others. This is important and relevant because as the popularity and importance of data science grows in all areas of society, the role of data analysis will be highlighted. The benefit of developing a formal characterization of the analytic process is a greater ability to scale training as well as an opportunity to make advances and develop tools to a improve the general quality of data analysis in all areas. 

Our work is related to the concepts of Bayesian and frequentist inference in that typical presentations of those paradigms only allow for a single iteration through the data. Whether we are computing a posterior distribution or calculating an hypothesis test, both paradigms take a single pass through the data to produce a result. There is typically no discussion of what happens next once the result is observed because of a lack of context. Therefore there is no real basis for initiating a second iteration. However, should a second iteration be initiated for whatever reason, our framework allows for the analyst to take either a Bayesian or frequentist approach in the next pass through the data. Hence, our approach here is entirely complementary to existing paradigms for inference.

The framework we have presented here is a first step towards characterizing the data analytic process that admits a number of new directions to explore. One such direction involves using the expected information gain and anomaly information gain criteria defined in Section~\ref{sec:analyticstep} to choose between tools in an analytic step. We outline how that might work in Supplemental Note~\ref{sec:choosingtools}. Another direction involves the process by which we set up the next iteration. The end of each iteration produces a contrast between the analyst's expectations and the observed output. The analyst must then react to that contrast and generate new ideas to pursue with the data. While a formal framework for how to think about that process is not yet clear, an obvious starting point would be to consider whether the output of the iteration is unexpected or as-expected. Unexpected results, in particular, may require an investigation within the data to determine their cause, and developing tools to accelerate this process could be an interesting area of future work. Another direction to explore could examine the minimal requirements for stopping the iterative process. In practice, analytic iterations may stop for a variety of reasons, but it may be possible to define a more formal stopping criterion to ensure that specific goals are achieved or a minimal amount of information is gained.

The area of data analysis is lacking in knowledge that can be applied in a range of applications and analytic situations often have unique circumstances. While our model of analytic iteration is just that---a model, the benefit of having some approximation to this process is that it allows for the development of new tools, approaches, or advice that can potentially be applied broadly. 



\clearpage







\bibliographystyle{agsm}
\bibliography{refs}

@article{lovett2000applying,
  title={Applying cognitive theory to statistics instruction},
  author={Lovett, Marsha C and Greenhouse, Joel B},
  journal={The American Statistician},
  volume={54},
  number={3},
  pages={196--206},
  year={2000},
  publisher={Taylor \& Francis}
}

@article{peng2017commenton,
author = {Roger D. Peng},
title = {Comment on “50 Years of Data Science”},
journal = {Journal of Computational and Graphical Statistics},
volume = {26},
number = {4},
pages = {767-767},
year  = {2017},
publisher = {Taylor & Francis},
doi = {10.1080/10618600.2017.1385470},
}

@article{wild1994embracing,
  title={Embracing the “wider view” of statistics},
  author={Wild, Chris J},
  journal={The American Statistician},
  volume={48},
  number={2},
  pages={163--171},
  year={1994},
  publisher={Taylor \& Francis}
}

@article{unwin2001patterns,
  title={Patterns of data analysis},
  author={Unwin, Antony},
  journal={Journal of the Korean Statistical Society},
  volume={30},
  number={2},
  pages={219--230},
  year={2001},
  publisher={Citeseer}
}

@article{greenhouse2018teaching,
  title={On teaching statistical practice: From novice to expert},
  author={Greenhouse, Joel B and Seltman, Howard J},
  journal={The American Statistician},
  volume={72},
  number={2},
  pages={147--154},
  year={2018},
  publisher={Taylor \& Francis}
}

@article{breznau2022observing,
  title={Observing many researchers using the same data and hypothesis reveals a hidden universe of uncertainty},
  author={Breznau, Nate and Rinke, Eike Mark and Wuttke, Alexander and Nguyen, Hung HV and Adem, Muna and Adriaans, Jule and Alvarez-Benjumea, Amalia and Andersen, Henrik K and Auer, Daniel and Azevedo, Flavio and others},
  journal={Proceedings of the National Academy of Sciences},
  volume={119},
  number={44},
  pages={e2203150119},
  year={2022},
  publisher={National Acad Sciences}
}

@book{tukey1977exploratory,
  title={Exploratory Data Analysis},
  author={Tukey, John W},
  volume={2},
  year={1977},
  publisher={Reading, MA}
}

@article{peng2021diagnosing,
  title={Diagnosing Data Analytic Problems in the Classroom},
  author={Peng, Roger D and Chen, Athena and Bridgeford, Eric and Leek, Jeffrey T and Hicks, Stephanie C},
  journal={Journal of Statistics and Data Science Education},
  volume={29},
  number={3},
  pages={267--276},
  year={2021},
  publisher={Taylor \& Francis}
}

@article{peng2022perspective,
  title={Perspective on data science},
  author={Peng, Roger D and Parker, Hilary S},
  journal={Annual Review of Statistics and Its Application},
  volume={9},
  pages={1--20},
  year={2022},
  publisher={Annual Reviews}
}

@article{grolemund2014cognitive,
  title={A cognitive interpretation of data analysis},
  author={Grolemund, Garrett and Wickham, Hadley},
  journal={International Statistical Review},
  volume={82},
  number={2},
  pages={184--204},
  year={2014},
  publisher={Wiley Online Library}
}

@article{wild1999statistical,
  title={Statistical thinking in empirical enquiry},
  author={Wild, Chris J and Pfannkuch, Maxine},
  journal={International statistical review},
  volume={67},
  number={3},
  pages={223--248},
  year={1999},
  publisher={Wiley Online Library}
}

@article{leek2010tackling,
  title={Tackling the widespread and critical impact of batch effects in high-throughput data},
  author={Leek, Jeffrey T and Scharpf, Robert B and Bravo, H{\'e}ctor Corrada and Simcha, David and Langmead, Benjamin and Johnson, W Evan and Geman, Donald and Baggerly, Keith and Irizarry, Rafael A},
  journal={Nature Reviews Genetics},
  volume={11},
  number={10},
  pages={733--739},
  year={2010},
  publisher={Nature Publishing Group UK London}
}

@article{anscombe1973graphs,
  title={Graphs in statistical analysis},
  author={Anscombe, Francis J},
  journal={The american statistician},
  volume={27},
  number={1},
  pages={17--21},
  year={1973},
  publisher={Taylor \& Francis}
}

@article{kimhardin2021,
author = {Albert Y. Kim and Johanna Hardin},
title = {“{Playing} the Whole Game”: A Data Collection and Analysis Exercise With Google Calendar},
journal = {Journal of Statistics and Data Science Education},
volume = {29},
number = {sup1},
pages = {S51-S60},
year  = {2021},
publisher = {Taylor & Francis},
doi = {10.1080/10691898.2020.1799728}
}

@article{shannon,
  author = {Shannon, Claude Elwood},
  journal = {The Bell System Technical Journal},
  pages = {379--423},
  title = {A Mathematical Theory of Communication},
  url = {http://plan9.bell-labs.com/cm/ms/what/shannonday/shannon1948.pdf},
  urldate = {2003-04-22},
  volume = 27,
  year = 1948
}

\clearpage

\appendix

\clearpage
{\huge \noindent Appendix}

\hrule

\vspace*{0.75cm}


{\Large \noindent Modeling Data Analytic Iteration With Probabilistic Outcome Sets}

\vspace*{0.75cm}

{\large \noindent Roger D. Peng, Stephanie C. Hicks}

\vspace*{0.5cm}



\section{Supplemental Notes}

\subsection{Proof of Theorem~\ref{theorem:different-info}}
\label{sec:suppnote1}


\begin{proof}
In this setting with two analysts, there are four possible scenarios when $\bY$ is observed: 
\begin{enumerate}
\item $\{\bY\in\epos_A\}\cap\{\bY\in\epos_B^c\}$, analyst $A$ as-expected, analyst $B$ unexpected;
\item $\{\bY\in\epos_A^c\}\cap\{\bY\in\epos_B\}$, analyst $A$ unexpected, analyst $B$ as-expected;
\item $\{\bY\in\epos_A\}\cap\{\bY\in\epos_B\}$, both analysts as-expected;
\item $\{\bY\in\epos_A^c\}\cap\{\bY\in\epos_B^c\}$, both analysts unexpected.
\end{enumerate}
If we consider Scenario~1, we know from Definition~\ref{def:exp-pot-outcome} that any expected potential outcome set must have a probability assigned by the analyst of greater than $0.5$. Also, any anomaly potential outcome set must have probability less than $0.5$. Therefore, it must have been that $\Phat_A(\boldy\in\epos_A) > 0.5$ and that $\Phat_B(\boldy\in\epos_B^c) < 0.5$, which implies that $\Phat_A(\boldy\in\epos_A)\ne\Phat_B(\boldy\in\epos_B^c)$. In this case, we know that for the observed information, $-\log\Phat_A(\boldy\in\epos_A)\ne-\log\Phat_B(\boldy\in\epos_B^c)$ because the function $-\log$ is strictly decreasing on $(0,1)$. A similar argument can be made for Scenario~2. For Scenario~3, we assumed that $\Phat_A(\boldy\in\epos_A)\ne\Phat_B(\boldy\in\epos_B)$. Therefore, should Scenario~3 occur, the information gain experienced by analyst $A$ and $B$ will be $-\log\Phat_A(\boldy\in\epos_A)$ and $-\log\Phat_B(\boldy\in\epos_B)$, respectively, which cannot be equal under the assumption. Finally, for Scenario~4, the assumptions imply that $\Phat_A(\boldy\in\epos_A^c)\ne\Phat_B(\boldy\in\epos_B^c)$ and therefore the information gain for the two analysts must also not be equal. Because in each of the four possible scenarios the probability of observing that outcome differs between the two analysts, the information gain also differs between the analysts.
\end{proof}

\subsection{Proof of Corollary~\ref{corollary:different-info}}
\label{sec:suppnote2}

\begin{proof}
In the proof of Theorem~\ref{theorem:different-info} above, we noted that there are four possible scenarios upon observing the output~$\bY$. If we know that the analysts experienced the same information gain, then Scenarios~1 and~2 are not possible since those scenarios cannot result in identical information gain. In other words, if the observed output~$\bY$ is as-expected for one analyst and unexpected for the other analyst, they cannot have experienced the same information gain. Under Scenario~3, if both analysts experienced the same information gain, that implies that $-\log\Phat_A(\boldy\in\epos_A)=-\log\Phat_B(\boldy\in\epos_B)$. It therefore follows that $\Phat_A(\boldy\in\epos_A)=\Phat_B(\boldy\in\epos_B)$. Under Scenario~4, if both analysts experienced the same information gain, that implies that $-\log\Phat_A(\boldy\in\epos_A^c)=-\log\Phat_B(\boldy\in\epos_B^c)$. It therefore follows that $\Phat_A(\boldy\in\epos_A^c)=\Phat_B(\boldy\in\epos_B^c)$, which implies that $\Phat_A(\boldy\in\epos_A)=\Phat_B(\boldy\in\epos_B)$.
\end{proof}

\subsection{Choosing Between Tools in Analytic Iteration}
\label{sec:choosingtools}

A fundamental problem that must be addressed at each step of an analytic iteration is the question of what tool to apply to the data. This question may be driven by the output from the previous iteration, which might have been as-expected or unexpected. Assume that for the next step the analyst can choose from a collection of tools informative $T_j$, $j=1,\dots,J$ to apply to the data $\bX$ in order to produce potential outcomes $\boldy_j=T_j(\bX)$. Given $T_j$ and $\bX$, there is an associated complete potential outcome set $\cpos_j$ in which $\boldy_j$ can fall. Furthermore, the analyst can construct expected and anomaly potential outcome sets $\epos_j$ and $\apos_j$ based on their own knowledge of $\bX$ and $T_j$. The choice the analyst must then make is between the available triples $(T_j,\epos_j,\apos_j)$, $j=1,\dots,J$. 

In the previous section, we proposed two criteria---expected information gain $\widehat{H}$, and maximum potential information gain $\widehat{M}$---as metrics for characterizing the potential information gain from taking an analytic step. We suggest here that these criteria could be used to prioritize different triples $(T_j,\epos_j,\apos_j)$ before observing the outcome when analyzing data in practice. Here, we describe how these criteria might be applied and how the different criteria may result in different choices.

Let $\widehat{H}_j$ ($j=1,\dots,J)$ be the expected information gain associated with the triple $(T_j,\epos_j,\apos_j)$. This is the expected information gained from applying $T_j$ to the dataset $\bX$ with the expected and anomaly potential outcome sets $\epos_j$ and $\apos_j$. If $\Phat$ is the analyst's probability model for the data, then we have
\[
\widehat{H}_j
=
\Phat(\boldy_j\in\epos_j)\left[-\log\Phat(\boldy_j\in\epos_j)\right]
+
\Phat(\boldy_j\in\apos_j)\left[-\log\Phat(\boldy_j\in\apos_j)\right].
\]
Given $\Phat$ and the triple $(T_j,\epos_j,\apos_j)$, the expected information gain $\widehat{H}_j$ can be computed before applying $T_j$ to the data. Therefore, each of the $J$ triples $(T_j,\epos_j,\apos_j)$ can be evaluated based on $\widehat{H}_j$ and the analyst could choose the triple $(T_j,\epos_j,\apos_j)$ that maximizes $\widehat{H}_j$ over $j=1,\dots,J$. If the anomaly information gain criteria were used, then we could similarly define $\widehat{M}_j$ as the anomaly information gain for triple $(T_j,\epos_j,\apos_j)$, i.e.
$\widehat{M}_j=-\log\Phat(\boldy_j\in\apos_j).$
The analyst could then choose the triple $(T_j,\epos_j,\apos_j)$ that maximizes $\widehat{M}_j$ over $j=1,\dots,J$.

The definitions of $\widehat{H}_j$ and $\widehat{M}_j$ suggest that each criterion is likely to favor different methods to apply to the data. The expected information gain criterion ($\widehat{H}_j$) will favor method triples that have $\Phat(\boldy_j\in\epos_j)$ close to $0.5$ because $\widehat{H}_j$ is maximized when $\Phat(\boldy_j\in\epos_j)$ is exactly equal to 0.5~(although that is not allowed given the definition of $\epos$). The anomaly information gain ($\widehat{M}_j$) will favor triples where $\Phat(\boldy_j\in\epos_j)$ is very near $1$, implying that $\Phat(\boldy_j\in\apos_j)$ is near zero and $-\log\Phat(\boldy_j\in\apos_j)$ is large. In the next section, we attempt to characterize common data analysis scenarios in which the different criteria might be more relevant for choosing an analytic tool.

\end{document}